\documentclass{article}
\usepackage{ijcai19}

\usepackage{times}
\usepackage{soul}
\usepackage{url}
\usepackage[hidelinks]{hyperref}
\usepackage[utf8]{inputenc}
\usepackage[small]{caption}
\usepackage{graphicx}
\usepackage{amsmath}
\usepackage{booktabs}
\usepackage{algorithm}
\usepackage{algorithmic}
\usepackage{amssymb}
\usepackage{amsthm}
\usepackage{xcolor}

\usepackage{latexsym}

\usepackage{etoolbox}
\apptocmd{\sloppy}{\hbadness 10000\relax}{}{}

\newtheorem{theorem}{Theorem}
\newtheorem{proposition}{Proposition}

\theoremstyle{definition}
\newtheorem{definition}{Definition}
\newtheorem{remark}{Remark}

\newcommand{\MoV}{\textsf{MoV}}
\newcommand{\Expec}{\mathbb{E}}

\title{Election Manipulation on Social Networks with Messages on Multiple Candidates}

\author{
	Matteo Castiglioni\textsuperscript{\textnormal 1}, Diodato Ferraioli\textsuperscript{\textnormal 2}, Giulia Landriani\textsuperscript{\textnormal 1}, Nicola Gatti\textsuperscript{\textnormal 1} \\
	\textsuperscript{1} \textnormal{Politecnico di Milano, Piazza Leonardo da Vinci 32, Milano, Italy} \\
	\textsuperscript{2} \textnormal{Universit\`a degli Studi di Salerno, Via Giovanni Paolo II, Fisciano, Italy }\\
	\textnormal{matteo.castiglioni@polimi.it, dferraioli@unisa.it, giulia.landriani@mail.polimi.it, nicola.gatti@polimi.it}
}

\begin{document} 

\maketitle

\begin{abstract}
We study the problem of election control through social influence when the manipulator is allowed to use the locations that she acquired on the network for sending \emph{both} positive and negative messages on \emph{multiple} candidates, widely extending the previous results available in the literature that study the influence of a single message on a single candidate. In particular, we provide a tight characterization of the settings in which the maximum increase in the margin of victory can be efficiently approximated and of those in which any approximation turns out to be impossible. We also show that, in simple networks, a large class of algorithms, mainly including all approaches recently adopted for social-influence problems, fail to compute a bounded approximation even on very simple networks, as undirected graphs with every node having a degree at most two or directed trees. Finally, we investigate various extensions and generalizations of the model.
\end{abstract}

\section{Introduction}
Nowadays, there is increasing use of social networks to convey inaccurate and unverified information, e.g., hosting the diffusion of fake news, spread by malicious users for their illicit goals. This can lead to severe and undesired consequences, as widespread panic, libelous campaigns, and conspiracies. In the United States, for instance, there is an ongoing discussion on the power of manipulation of social media during the US elections in 2016~\cite{USelection} and, more importantly, for future elections. Thus, understanding and limiting  the adverse effects on the elections due to information diffusion in social networks is currently considered of paramount importance.

The problem of election control through social influence has been recently the object of interest of many works. E.g., \citeauthor{sina2015adapting}~[\citeyear{sina2015adapting}] show how to modify the relationship among voters in order to make the desired candidate to win an election; \citeauthor{auletta2015minority}~[\citeyear{auletta2015minority,auletta2017information,auletta2017robustness}] show that, in case of two only candidates, a manipulator controlling the order in which information is disclosed to voters can lead the minority to become a majority;  a similar adversary is studied by \citeauthor{auletta2018reasoning}~[\citeyear{auletta2018reasoning}], showing that such a manipulator can lead a bare majority to consensus; \citeauthor{bredereck2017manipulating}~[\citeyear{bredereck2017manipulating}] study how selecting seeds from which to diffuse information to manipulate a two-candidate election.

In our paper, we focus on this last kind of manipulation.
More precisely, we assume that a manipulator can buy some locations on the network from which she kicks off a diffusion of (potentially fake) news aiming at altering the preference rankings of the voters receiving them to make the desired candidate to win the election.
\citeauthor{wilder2018controlling}~[\citeyear{wilder2018controlling}] have recently studied the case in which a manipulator spreads information on a single candidate to make her win the election or lose the election. The authors also provided approximation algorithms to compute the optimal seeds for positive messages only or negative messages only.

However, in elections with more than two candidates, the assumption that the election control can be either constructive or destructive as well as on a single candidate only is too limiting. Consider, for example, the setting in Figure~\ref{fig:clique}, characterized by five voters and five candidates, with $c_0$ denoting the manipulator's candidate and $c_1, \ldots, c_4$ denoting the remaining candidates. Among the five voters, we assume that one of them already prefers the manipulator's candidate $c_0$ to $c_4$ and these two to the remaining candidates;
the preference of the remaining voters are arranged as follows:
one voter prefers $c_1$ to $c_0$ and these two to the remaining candidates;
one voter prefers $c_1$ to $c_3$ and these two to the remaining candidates;
one voter prefers $c_2$ to $c_3$ and these two to the remaining candidates;
one voter prefers $c_2$ to $c_4$ and these two to the remaining candidates.
We assume that voters are arranged on a clique and each information is received with probability one regardless of the sender. Hence, we do not need to care about which nodes send positive or negative messages since all voters will always receive these messages. It can be easily observed that when a message (positive or negative) on a single candidate is sent, then the desired candidate $c_0$ cannot be made to win the election (at most, the election will end with a tie between $c_0$ and another candidate, both taking two votes). Instead,  injecting the network with a positive message toward $c_0$ and a negative message toward $c_2$ will result in $c_0$ being the only node with two votes, and thus the winner. This example also shows that, differently from what happens in the setting studied by~\citeauthor{wilder2018controlling}~[\citeyear{wilder2018controlling}], the optimal solution can include positive/negative messages on candidates different from the desired one.

\begin{figure}
	\includegraphics[width=\linewidth]{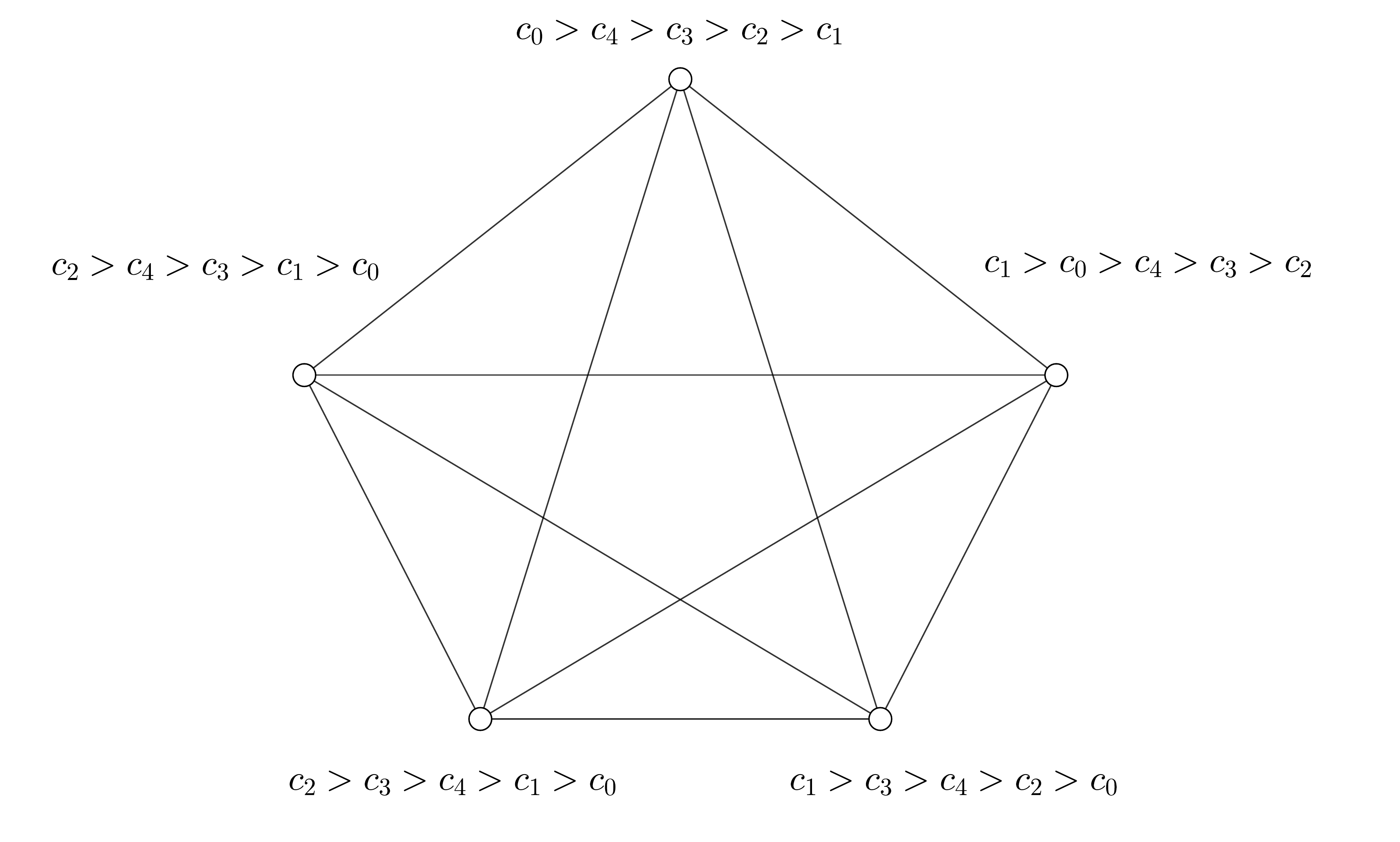}
	\caption{Clique characterized by five voters and five candidates.}
	\label{fig:clique}
\end{figure}

\paragraph{Our Contribution.}
We focus on the election control problem when the manipulator is allowed to use the locations that she acquired on the network for sending \emph{both} positive and negative messages on \emph{multiple} candidates. Our model extends the one by \citeauthor{wilder2018controlling}~[\citeyear{wilder2018controlling}] and assumes that diffusion occurs according to a variation of the \emph{independent cascade model} \cite{kempe2003maximizing} able to capture the simultaneous spread of multiple messages. We also focus on the same objective studied by~\citeauthor{wilder2018controlling}~[\citeyear{wilder2018controlling}], that is to maximize the increase in the margin of victory of the manipulator's candidate $c_0$.

Under mild assumptions on how a voter revises her preference ranking given a set of messages, we provide a tight characterization of the settings in which the maximum increase in the margin of victory can be efficiently approximated and of those in which any approximation turns out to be impossible unless $\mathsf{P} = \mathsf{NP}$. Specifically, we prove that whenever there is a set of $\tau$ messages making the candidate initially ranked by a voter as the least-preferred one to the most-preferred candidate, then there is a greedy poly-time algorithm guaranteeing an approximation factor $\rho$ depending on $\tau$. A surprisingly sharp transition phase occurs, instead, when no such a set of messages exists. In this case, no poly-time approximation algorithm is possible, unless $\mathsf{P} = \mathsf{NP}$, even when the approximation factor is a function in the size of the problem.

This last result poses a severe obstacle to the possibility for a manipulator to successfully alter the outcome of an election. Even more importantly, we show that this hardness result does not hold merely for worst-case (thus, potentially, knife-edge or rare) instances. Indeed, we prove that a large class of algorithms, that mainly include all approaches that have been recently adopted for social-influence problems, fail to compute  an empirically bounded approximation even on very simple networks, as undirected graphs with every node having a degree at most two or directed trees.
We also observe that, as a corollary of our characterization, the election control problem is inapproximable within any factor when positive only influence or negative only influence is possible. Let us remark that this inapproximability result is dramatically different from that obtained by~\citeauthor{wilder2018controlling}~[\citeyear{wilder2018controlling}], who show that, when the manipulator spreads only positive or only negative information on a single candidate, a constant approximation can be achieved.

Finally, we study some variants and generalizations of our model. More precisely, at first, we investigate the case in which seeds are \emph{bribed}, i.e., for each seed her ranking (and thus, her vote) is not affected by messages different from the one that she sends. Second, we study the case in which one uses alternative objective functions such as, e.g.,  the probability of victory. Third, we exmplore the case in which the diffusion occurs according to the \emph{linear threshold model}~\cite{kempe2003maximizing} and, finally,  the case in which different seeds may have a different cost for the manipulator. We show that our complexity results also hold in these variants.

\section{The Model}
\label{sec:model}
We consider an \emph{election control problem}, defined by a set of candidates $C=\{c_0,c_1,\dots,c_\ell\}$,
and a network of voters, modeled as a weighted directed graph $G = (V,E, p)$, where $V$ is the set of voters, $E$ is the set of edges, and $p \colon E \rightarrow (0,1]$ denotes the strength of the influence among voters. In particular, with a slight abuse of notation, we denote as $p(u,v)$ the strength of influence of $u$ on $v$.
Each voter $v$ has a preference ranking $\pi_v$ over the candidates.
We denote as $\pi_v(i)$ the $i$-th candidate in the rank $\pi_v$.
At the election time, the voter is assumed to cast a vote for $\pi_v(1)$.
For each candidate $c \in C$, we also denote as $V_c$ the set of voters that rank $c$ as first, i.e., $V_c = \left\{v \in V \mid \pi_v (1) = c\right\}$.

The election control problem involves a single  agent (i.e., the manipulator) whose objective is to spend a \emph{budget} $B$ of messages
to make $c_0$ win the election, by injecting in the network positive or negative information both 
about $c_0$ and about the other candidates.
Namely, our goal is to find a set of \emph{seeds} of $V$
and a set of at most $B$ messages sent by these seeds in order to maximize the increase in the margin of victory of $c_0$.

Specifically, let $S \subseteq V$ be a subset of voters and $I(s)=(q_0,...,q_{\ell})$ be
a vector  associated to each $s \in S$,
where $q_i \in \{-, \cdot,+\}$, with $q_i = +$ ($q_i = -$, respectively)
representing that $s$ sends a positive (negative, respectively) message about candidate $c_i$,
and $q_i = \cdot$ representing that no message is sent by $s$ about $c_i$.
For every $s \in S$, given a vector $I(s)$, we denote as $|I(s)|$ the number of messages $+$ or $-$ sent by seed $s$, i.e., $|I(s)| = |\{i \colon q_i \neq \cdot\}|$.
We assume that $|I(s)| \geq 1$ for every $s \in S$.
We also say that $s$ sends \emph{message} $(c,q)$ for $c \in C$ and $q \in \{+,-\}$ if $I(s,c)$, i.e., the $c$-th entry of $I(s)$, is $q$.
Given $(S,I)$, its cost 
is defined as the cumulative number of messages sent by the seeds,
i.e., 
$\sum_{s \in S} |I(s)|$.
A solution $(S, I)$ is \emph{feasible} does not exceed cost is 
$\leq B$.

For each feasible solution $(S, I)$, messages are supposed to spread over the network according to a \emph{multi-issue independent cascade} (MI-IC) model. In this model, given the graph $G = (V, E, p)$, we define the \emph{live-graph} $H = (V, E')$, where each edge $(u,v) \in E$ is included in $H$ with probability $p(u,v)$.
In this model, for each candidate $c \in C$ and for each type $q \in \{-,+\}$, we keep a set $A^t_{c,q}$ of \emph{active} voters at time $t$.
These sets initially contain the seeds sending the corresponding messages,
i.e. $A^0_{c,q} = \{s \in S \colon I(s,c) = q\}$ for every $c, q$.
Finally, at each time $t \geq 1$, we build $A^t_{c,q}$ as follows:
for each edge $(s,v) \in E'$, we consider the set $M(s,v)$ of messages $(c,q)$
such that $s \in A^{t-1}_{c,q}$ and $v \notin \bigcup_{i < t} A^{i}_{c,q}$;
then for each $(s,v)$ such that $M(s,v)$ is not empty,
we add $v$ to $A^t_{c,q}$ for every $(c,q) \in M(s,v)$.
The diffusion process of the message $(c,q)$ terminates at time  $T_{c,q}$ such that $A^{T_{c,q}}_{c,q} = \emptyset$.
Finally, the cascade terminates when the diffusion of each message $(c,q)$ terminates.

Roughly speaking, this process models the following realistic behavior:
each seed $s$ propagates all messages in $I(s)$ to neighbors;
however, a voter $v$ receiving messages from $s$ may not accept the information that these messages carry:
the acceptance probability indeed depends on the strength of the influence that $s$ has on $v$,
and hence, we say that $v$ will be \emph{activated} by $s$ only with probability $p(s,v)$ and this corresponds to $(s,v)$ being an edge of the live-graph $H$;
finally, each newly activated voter tries to influence the vote of their neighbors that have not yet accepted the spreading information, and to this aim they simply forward the messages they received. That is, freshly convinced voters act as new seeds, and the process continues as long as there is some voter willing to play the role of the seed.
Note that we assume that each node simply forwards all messages that she receives, but each node processes a given message only once.

The reception by voter $v$ of messages and the acceptance of their content do not only influence whether $v$ will or not forward these messages through the network, but also affect her preference ranking. Denote with $R = \{(c,q)\}_{c \in C}$ a set of received messages. A \emph{ranking revision function} $\phi$ associates each pair $(\pi,R)$ a new ranking $\pi'$ obtained by revising ranking $\pi$ with the set of received messages $R$. We study a general class of ranking revision functions, described below, that extends that one used by~\citeauthor{wilder2018controlling}~[\citeyear{wilder2018controlling}].

Given a feasible solution $(S,I)$ and a live graph $H$, we let, for every $v \in V$, $\pi^*_v(S,I,H)$ be the ranking at the end of the MI-IC model.
Moreover, for each candidate $c \in C$, we also denote as $V^*_c$ the set of voters that rank $c$ as first at the end of the diffusion process, i.e., $V^*_c(S,I,H) = \left\{v \in V \mid \pi^*_v (1) = c\right\}$. We finally let the \emph{margin of victory} of $(S, I, H)$ to be
$$\MoV(S,I,H) = |V^*_{c_0}(S,I,H)| - \max_{c \neq c_0} |V^*_{c}(S,I,H)|,$$
that denotes the number of votes that $c_0$ needs to win the election, if the first term is lower than the second, and the advantage of $c_0$ with respect to the second best ranked candidate, otherwise.
Finally, the \emph{effectiveness} of $(S, I)$, denoted as $\Delta_\MoV(S,I,H)$ is given by the increase in the margin of victory due to this choice of seeds and messages, i.e.,
$$\Delta_\MoV(S,I,H) = \MoV(S,I,H) - \MoV(\emptyset,(),H).$$

Hence, the election control problem consists in computing $(S^*, I^*) = \arg \max_{(S,I)} \Expec_H[\Delta_\MoV(S,I,H)]$,
where expectation is taken on the probability of live graphs $H$.

An algorithm $A$ is said to always return a $\rho$-approximation for the election control problem with $\rho \in [0,1]$ potentially depending on the size of the problem, if, for each instance of the problem,
it returns $(S, I)$ such that $\Expec_H[\Delta_\MoV(S,I,H)] \geq \rho\, \Expec_H[\Delta_\MoV(S^*,I^*,H)]$.

\paragraph{Ranking Revision Functions.} We consider a general class of ranking revision functions $\phi$ defined as follows. When there is a single message on the network, the ranking revision is as the one prescribed by~\citeauthor{wilder2018controlling}~[\citeyear{wilder2018controlling}]. That is, a message $(c,+)$ causes that $c$ switches her position with the candidate above, whereas each message $(c,-)$ causes that $c$ switches her position with the candidate below. Instead, if $v$ receives both $(c,+)$ and $(c,-)$, then, she discards them and behaves as if no message was received about $c$. When there are messages on multiple candidates, the ranking revision functions  $\phi$ satisfy the following mild properties.
\begin{itemize}
\item Be given a ranking $\pi$ and two message sets $R, R'$, differing only for a single candidate, say $c_i$, such that $R'$ contains $(c_i, -)$, while $R$ contains $(c_i, +)$ or $(c_i, \cdot)$. If $c_j\neq c_i$ is the most-preferred candidate of the ranking returned by $\phi(\pi,R)$, then $c_j$ must also be the most-preferred candidate of the ranking returned by $\phi(\pi,R')$. This is equivalent to say that, if candidate $c_j$ is the most preferred, then she keeps to be the most preferred when an additional negative message on an alternative candidate is received. 
\item Be given two possible rankings $\pi,\pi'$ of a node $v$ that differ only for the position of candidate $c_j$, in $\pi$ not being worse than in $\pi'$. Be given a message set $R$. If $c_j$ is the most preferred candidate in the ranking returned by $\phi(\pi',R)$, then $c_j$ must be the most preferred also in the ranking returned by $\phi(\pi,R)$.
\end{itemize}

Some examples of ranking revision functions may be provided, e.g., based on different orderings with which the single messages of $R$ are applied to switch candidates in the ranking or based on scoring rules. Consider the first case and, for the sake of presentation, focus on only three candidates. Consider a voter $v$ with $\pi_v = c_0 \succ c_1 \succ c_2$. 
We can have different resulting rankings for the same set of messages, e.g.:
\begin{itemize}
 \item $(\cdot,-,+)$, that can result in $\pi^*_v = c_2 \succ c_0 \succ c_1$, if we first apply the change induced by message $(c_1,-)$, or in  $\pi^*_v = c_0 \succ c_2 \succ c_1$, otherwise;
 \item $(-,-,\cdot)$, that can result in  $\pi^*_v = c_0 \succ c_1 \succ c_2$, if we first apply the change induced by message $(c_0,-)$, or in  $\pi^*_v = c_2 \succ c_0 \succ c_1$, otherwise.
\end{itemize}
Among all the possible $\phi$, we focus on the most extremal rules: the \emph{pessimistic} ranking revision function assumes that every candidate not ranked in the first two positions can never become the most preferred candidate (thus when applied to the case of three candidates as described above, it implies that ties are always broken in favor of the ranking $\pi^*_v$ such that $\pi^*_v(1) \neq c_2$); the \emph{optimistic} ranking revision function instead does exactly the opposite and always breaks ties by favoring the worst ranked candidate. Now, we consider the case in which the ranking revision functions are based on scoring rules and we report an example named \emph{score-based}. Here, we assume that exchanged messages have a more prominent role in the decision about who $v$ will vote with respect to her initial ranking. In order to describe this function, it would be convenient to assume that voter $v$ assigns a score to each candidate based on her position in $\pi_v$: she assigns score $|C|$ to the first-ranked candidate, score $|C|-1$ to the second, and so on; moreover, for each candidate $c$, the message $(c,+)$ corresponds to increase its score by $1 + \varepsilon$, and the message $(c,-)$ corresponds to decrease this score by the same amount; the final ranking is then computed with respect to these updated scores. Notice that the score-based ranking revision function is well defined regardless the number of candidates.
Table~\ref{table:ties} summarizes the behavior of the pessimistic, optimistic, and score-based ranking revision functions when $|C|=3$.
\begin{table}[!ht]
	\small
	\centering
	\renewcommand{\arraystretch}{1.5}
	\begin{tabular}{|c|c|c|c|}
		\hline
							 		&  \textbf{pessimistic} 			&  \textbf{optimistic} 			&	\textbf{score-based}		\\ \hline
		$(+,\text{any},\text{any})$ 			& $c_0$ 						& $c_0$ 					& 		$c_0$ 			\\ \hline
		$(\cdot,+,\text{any})$ 			& $c_1$ 						& $c_1$ 					& 		$c_1$ 			\\ \hline
		$(\cdot,-,\text{any})$ 				& $c_0$ 						& $c_0$ 					& 		$c_0$ 			\\ \hline
		$(-,+,\text{any})$ 				& $c_1$ 						& $c_1$ 					& 		$c_1$ 			\\ \hline
		$(-,\cdot,+)$ 					& $c_1$ 						& $c_2$ 					& 		$c_2$ 			\\ \hline
		$(-,\cdot,-)$ 					& $c_1$ 						& $c_1$ 					& 		$c_1$ 			\\ \hline
		$(-,-,+)$ 						& $c_0$ 						& $c_2$ 					& 		$c_2$ 			\\ \hline
		$(-,-,\cdot)$ 					& $c_0$ 						& $c_2$ 					& 		$c_0$ 			\\ \hline
		$(-,-,-)$ 						& $c_0$ 						& $c_0$ 					& 		$c_0$ 			\\ \hline
	\end{tabular}
	\renewcommand{\arraystretch}{1}
	\caption{Most-preferred candidate in the pessimistic, optimistic, and score-based ranking revision functions with at least two messages.}
\label{table:ties}
\end{table}

\section{General Results}
Initially, we observe that for every ranking revision function, the following holds:
\begin{itemize}
\item given a node whose preferences are such that $c_0$ is the least-preferred candidate, either $c_0$ becomes the most-preferred candidate by the message set $\{(c_0,+),(c_i,-) \textnormal{ for every $i >0$ } \}$ or there is no $R$ such that $c_0$ can become the most-preferred candidate;
\item if the message set $\{(c_0,+),(c_i,-) \textnormal{ for every $i >0$ } \}$ makes candidate $c_0$ the most-preferred one for a given ranking $\pi$ when $c_0$ is the least-preferred candidate, then the same message makes $c_0$ the most-preferred candidate for any other ranking $\pi'$.
\end{itemize}
We introduce the following definition that we use to  capture the phase transition of the election control problem.
\begin{definition}
The pair $(\phi, |C|)$ composed of a ranking revision function and a number of candidates is said \emph{least-candidate manipulable} if, when $c_0$ is the least-preferred candidate for a node, message set $\{(c_0,+),(c_i,-) \textnormal{ for every $i >0$ } \}$ does always make $c_0$ be the most preferred.
\end{definition}

\subsection{Inapproximability Results}

\begin{theorem}
\label{thm:inapprox}
Be given the set of instances in which $(\phi, |C|)$ is not least-candidate manipulable. For any $\rho > 0$ even depending on the size of the problem, there is not any poly-time algorithm returning a $\rho$-approximation to the election control problem, unless $\mathsf{P} = \mathsf{NP}$.
\end{theorem}
\begin{proof}
The proof uses a reduction from the well-known $\mathsf{NP}$-hard problem \emph{set cover}. This problem, given a finite set $N=\{z_1,\dots,z_n\}$ of elements, a collection $X=\{x_1,..,x_m\}$ of sets with $x_i \subseteq N$, and an integer $h$, asks whether there is a collection $X^*\subseteq X$ such that $|X^*|\le h$ and $\cup_{x_i \in X^*}\ x_i= N$.

Given an instance of set cover, we build an instance of the election control problem with $\ell+1$ candidates as follows.
The voters' network $G$, showed in Figure~\ref{fig:reduction}, consists of four disconnected components, that we denote as $G_1, \ldots, G_4$.
Note that all edges of $G$ have $p(u,v) = 1$.

\begin{figure}
	\includegraphics[width=1\linewidth]{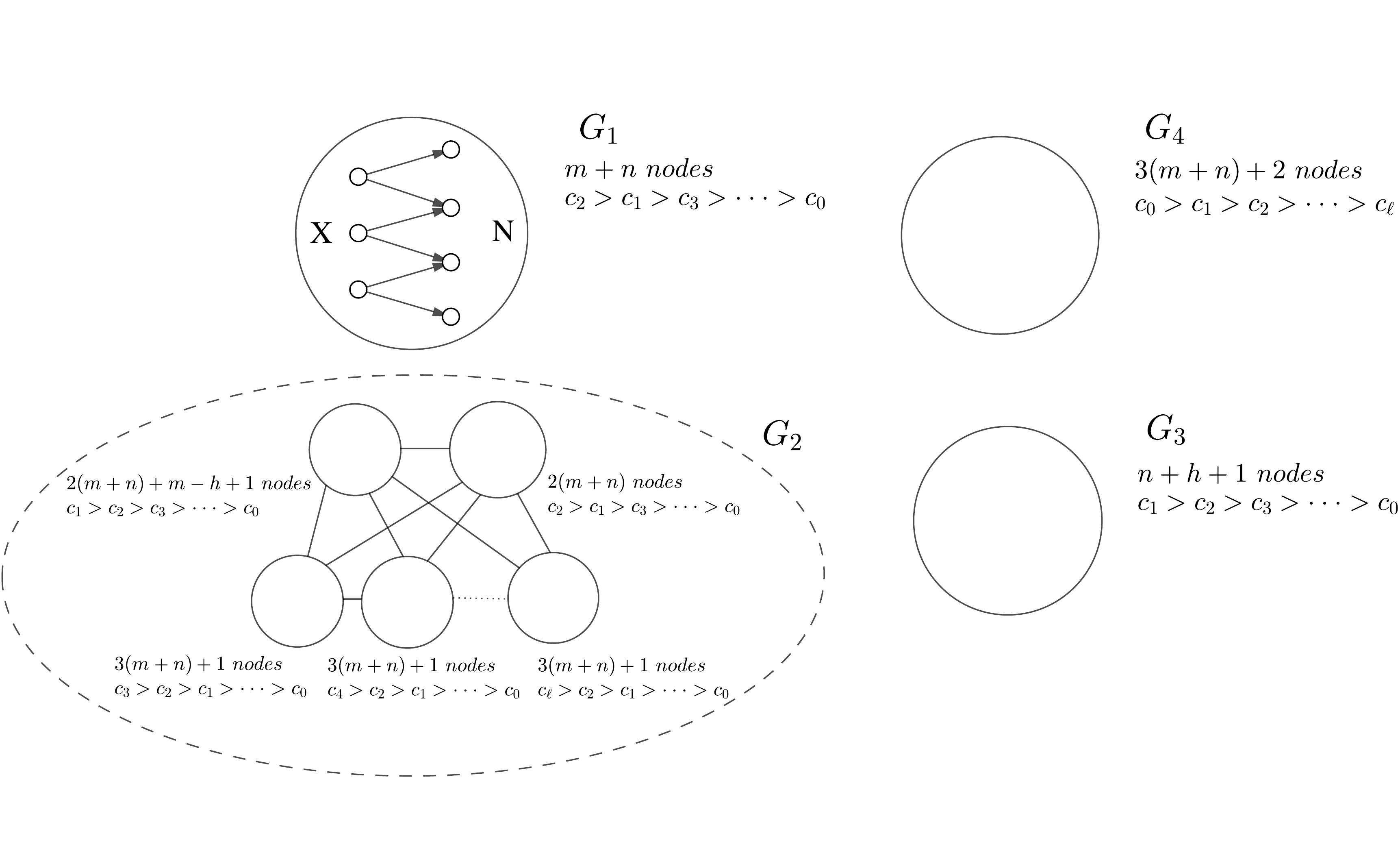}
	\caption{Reduction for an election scenario characterized by $\ell + 1$ candidates.}
	\label{fig:reduction}
\end{figure}

The component $G_1$ has $m+n$ nodes and it used to model the set cover instance. Indeed, for each $z_i \in N$, we have in $G_1$ a node $v_{z_i}$;
moreover, for each $x_i \in X$, we have in $G_1$ a node $v_{x_i}$
with an edge toward $v_z$ for each $z \in x_i$.
All voters $v$ corresponding to nodes in $G_1$ have ranking $\pi_{v}=c_2>c_1>c_3>\cdots>c_{\ell}>c_0$.

The component $G_2$ is a clique of $(3\ell-5)(m+n)+m+\ell-h-4$ nodes, such that
\begin{itemize}
\item $2(m+n)$ nodes have ranking $\pi_{v}=c_2>c_1>c_3>\ldots>c_{\ell}>c_0$;
\item $2(m+n)+m-h+1$ nodes have ranking $\pi_{v}=c_1>c_2>c_3>\ldots>c_{\ell}>c_0$;
\item for every $i \not \in \{0,1,2\}$, $3(m+n)+1$ nodes have ranking $\pi_{v}=c_i>c_2>c_1>c_3>\ldots>c_{i-1}>c_{i+1}>\ldots>c_{\ell}>c_0$.
\end{itemize}

The component $G_3$ is a clique of $n+h+1$ nodes, such that every node has ranking $\pi_{v}=c_1>c_2>c_3>\ldots>c_{\ell}>c_0$.

The component $G_4$ is a clique of $3(m+n)+2$ nodes, such that every node has ranking $\pi_{v}=c_0>c_1>c_2>c_3>\ldots>c_{\ell}$.

Note that $|V_{c_0}| = 3(m+n)+2$, $|V_{c_1}| = 3(m+n)+2$, $|V_{c_2}| = 3(m+n)$, and $|V_{c_i}| = 3(m+n)+1$ for every $i \neq 0, 1, 2$. Hence, $\MoV(\emptyset,(),H) = 0$.
Finally, we set the budget $B=h+1$.

We next prove that this instance allows a feasible solution $(S^*, I^*)$ with $\MoV(S^*,I^*,H) > 0$ if and only if there a set cover of size at most $h$.

\textbf{(If)} Let $X^* \subseteq X$ be the set cover of size $h$\footnote{If there is a set cover $X^*$ of size less $h$, then we can achieve a set cover of size exactly $h$, by padding $X^*$ with arbitrary element in $X \setminus X^*$.} (i.e., $|X^*| = h$ and $\cup_{x_i \in X^*}=N$).
 Then we set $(S^*, I^*)$ as follows: for every $x_i \in X^*$, we include $v_{x_i} \in S^*$ and we set $I^*(v_{x_i})$ such that $q_1 = +$, and $q_i = \cdot$ for every $i \neq 1$; moreover, we include in $S^*$ an arbitrary node $v \in G_3$ and we set $I^*(v)$ such that $q_2 = +$, and $q_i = \cdot$ for every $i \neq 2$.
 
 It directly follows that $(S^*, I^*)$ is feasible.
 We next show that $\MoV(S^*,I^*,H) > 0$.
 Indeed, it directly follows that $|V^*_{c_i}(S^*,I^*,H)| = |V_{c_i}|$ for every $i \neq 1,2$. Moreover, the diffusion of messages leads each voter corresponding to nodes in $G_3$ to prefer $c_2$ to $c_1$. Moreover, the dynamics leads $h + n$ voters in $G_1$ (i.e., the seeds and the ones corresponding to elements $z_i \in N$) to prefer $c_1$ to $c_2$.
 Hence, $|V^*_{c_1}(S^*,I^*,H)| = |V_{c_1}| - |G_3| + h+n = 3(m+n)+1$,
 and $|V^*_{c_2}(S^*,I^*,H)| = |V_{c_2}| + |G_3| - h-n = 3(m+n)+1$.
 Hence, $\MoV(S^*,I^*,H) = 1$, as desired.
 
\textbf{ (Only if)} Suppose that there exists a pair $(S^*, I^*)$ such that $\MoV(S^*,I^*,H) > 0$. Note that, since $c_0$ is ranked either as the first or as the last by all voters and the instance is not least-candidate manipulable, then $|V^*_{c_0}(S^*,I^*,H)| = |V_{c_0}|$. Hence, in order to have $\MoV(S^*,I^*,H) > 0$, it must be the case that the number of voters whose most-preferred candidate is $c_1$ decreases by at least one unit, the number of voters whose most-preferred candidate is $c_2$ increases by at most one unit, and the number of voters whose most-preferred candidate is $c_i$, for $i \not \in\{ 0, 1, 2\}$, does not increase.

 We first consider the clique $G_2$, proving two important properties.
 Suppose, that a candidate $c_i$ loses votes in favor of $c_j$.
 Then there is a voter receiving messages that allow $c_j$ to pass in the ranking candidate $c_i$. 
 However, since $G_2$ is a clique, all nodes receive these messages,
 and thus all votes of $c_i$ are taken by $c_j$.
 Moreover, no candidate can take the votes of another candidate without losing her initial votes, otherwise she would have at least $4(m+n)+m-h+1$ votes and $\MoV(S^*,I^*,H) \leq 0$.
 We can now prove that $c_2$ cannot change her voters in $G_2$.
 Suppose $c_2$ gains votes in $G_2$ and another candidate $c_j$ takes her votes. Notice that $c_j$ must lose her initial votes. Since the rankings of the voters of $c_2$ and $c_j$ differ only for the ranking of $c_j$ and all voters receive the same messages, the second assumption on the 
 ranking revision function (see previous section) is not satisfied.

 
 
 Then, $c_2$ must take voter in $G_3$. Again, since $G_3$ is a clique, it must be that all votes of $c_1$ are taken by candidate $c_2$.
 
 Thus, $c_1$ loses all its voters in $G_3$ in favor of $c_2$. Note that a single message is sufficient (a positive message for $c_2$) to this aim. However, this implies that $c_2$ must lose $n+h$ voters in $G_1$, otherwise $|V^*_{c_2}(S^*,I^*,H)| > 3(m+n)+n+h+1-(n+h)$ and thus $\MoV(S^*,I^*,H) \leq 0$, that contradicts our hypothesis. Observe that these votes must be necessarily lost in favor of $c_1$.
 
 Hence, we are left with $h$ available messages to make $n+h$ voters to change their vote from $c_2$ to $c_1$. Observe that, in order to make a voter to change, it is sufficient a single message (a positive message for $c_1$). However, if less than $h$ seeds sending this message are located among nodes $v_{x_i}$ for $x_i \in X$, then less than $n+h$ voters will change their mind (since nodes $v_{x_i}$ for $x_i \in X$ have no ingoing edges).
 
 Finally, we must have that the $h$ seeds in $G_1$ are neighbors of every node $v_{z_i}$ for $z_i \in N$. Hence, the set $X^* = \{x_i \colon v_{x_i} \in S^*\}$ has size $h$ and, by construction of $G_1$, $\bigcup_{x \in X^*} x = N$, i.e. $X^*$ is a set cover of size at most $h$.

Hence we can conclude that a feasible solution $(S^*,I^*)$ with $\Delta_{\MoV}(S^*,I^*,H) > 0$ exists if and only if a solution for set cover exists. Note also that if a solution with $\Delta_{\MoV}(S^*,I^*,H) > 0$ exists, then $\Delta_{\MoV}(S,I,H) > 0$ even for any $\rho$-approximate solution $(S,I)$, regardless of the value of $\rho$. Thus, if a polynomial time $\rho$-approximate algorithm for election control problem exists, then the set cover problem can also be solved in poly-time, implying that $\mathsf{P} = \mathsf{NP}$.
\end{proof}

\begin{remark}
 For sake of presentation, in the above proof we considered a disconnected graph. However, if we restrict $\rho$ to be upper bounded by an exponential function of the size of the problem, the proof can be immediately extended to strongly connected graphs, by adding in the graph described above $r$ edges to make it strongly connected.
 If we denote with $\varepsilon > 0$ the weight of these edges, we observe that there are $2^r - 1$ live graphs involving at least one of these edges, and each of them has probability at most $\varepsilon$. Moreover, in each of these live graphs $H$, it holds $-|N| \leq \MoV(S^*,I^*, H) \leq |N|$, where $N$ is the set of nodes in the above proof. Thus, by denoting with $H^*$ the live graph that does not involve any of these edges, we have that $\Expec_H \left[\MoV(S^*,I^*, H)\right] = \MoV(S^*,I^*, H^*) \cdot \Pr(H^*) + \sum_{H \neq H^*} \MoV(S^*,I^*, H) \cdot \Pr(H)$. If we assume $\MoV(S^*,I^*, H^*) = 1$, then
 $\Expec_H \left[\MoV(S^*,I^*, H)\right] \geq (1-(2^r-1)\varepsilon) - (2^r-1)|N|\varepsilon$.
 If, instead, $\MoV(S^*,I^*, H^*) = 0$, then $\Expec_H \left[\MoV(S^*,I^*, H)\right] \leq (2^r-1)|N|\varepsilon$.
 If $\varepsilon < \frac{\rho}{(2^r-1) \left((\rho+1)|N| + \rho\right)}$, we have that any $\rho$-approximate solution $(S', I')$ has $\Expec_H \left[\MoV(S',I', H)\right] > \frac{|N|}{(2^r-1) \left((\rho+1)|N| + \rho\right)}$ in the case that $\MoV(S^*,I^*, H^*) = 1$  and $\Expec_H \left[\MoV(S',I', H)\right] < \frac{\rho|N|}{(2^r-1) \left((\rho+1)|N| + \rho\right)}$ otherwise.
 
 Hence, any $\rho$-approximate equilibrium is able to distinguish whether $\MoV(S^*,I^*, H^*) = 1$ or not, and thus to solve the set cover problem.\qed
\end{remark}

Theorem~\ref{thm:inapprox} essentially states that whenever there is no way for making the least-preferred candidate for a node to become the most-preferred one of that node, then there is no chance that a manipulator designs an algorithm allowing her to maximize the increment in the margin of victory of the desired candidate regardless the adopted ranking revision function. However, Theorem~\ref{thm:inapprox} does not rule out that the worst-case instances on which the manipulator's algorithm fails are very rare and/or knife-edge. Even if this was the case, we show that, if the manipulator greedily chooses the messages to send, then her approach fails even for simple graphs, namely graphs with all nodes having a degree of $2$ or trees.

Specifically, given a set $S$ of seeds and corresponding messages $I$,
we denote as $\mathcal{F}(S,I)$ the set of pairs $(s, I(s))$, with $s \notin S$ such that either
 $\Expec_H \left[\MoV(S \cup \{s\},(I, I(s)), H)\right] > \Expec_H \left[\MoV(S,I, H)\right]$
 or $\Expec_H \left[V^*_{c_0}(S \cup \{s\},(I, I(s)), H)\right] > \Expec_H \left[V^*_{c_0}(S,I, H)\right]$. That is, $\mathcal{F}(S,I)$ includes all the ways of augmenting a current solution so that either the margin of victory of $c_0$  or the number of its votes increases.
 Then, we say that an algorithm to solve the election control problem uses the \emph{greedy} approach, if it works as follows:
\begin{itemize}
 \item it starts with $S=\emptyset$ and $I=()$;
 \item until the set $\mathcal{F}(S,I)$ is not empty, choose one $(s, I(s)) \in \mathcal{F}(S,I)$ and set $S = S \cup \{s\}$ and $I=(I,I(s))$.
\end{itemize}
We then show that every algorithm in this class fails even for very simple networks.

\begin{proposition}
\label{prop:example1}
Be given the set of instances in which $(\phi, |C|)$ is not least-candidate manipulable. For any $\rho > 0$ even depending on the size of the problem, no algorithm following the greedy approach  returns a $\rho$-approximation, even in undirected graphs in which each node has degree at most~$2$.
\end{proposition}
\begin{proof}
 Consider the graph given in Figure~\ref{fig:example1}.
 According to the preference rankings of the nodes, the candidate $c_2$ collects 5 votes, while candidates $c_1$ and $c_0$ gather 7 votes each. So the actual margin is equal to zero. Suppose the budget $B$ be~2.

\begin{figure}[!htb]
	\centering
	\includegraphics[width=1.1\columnwidth]{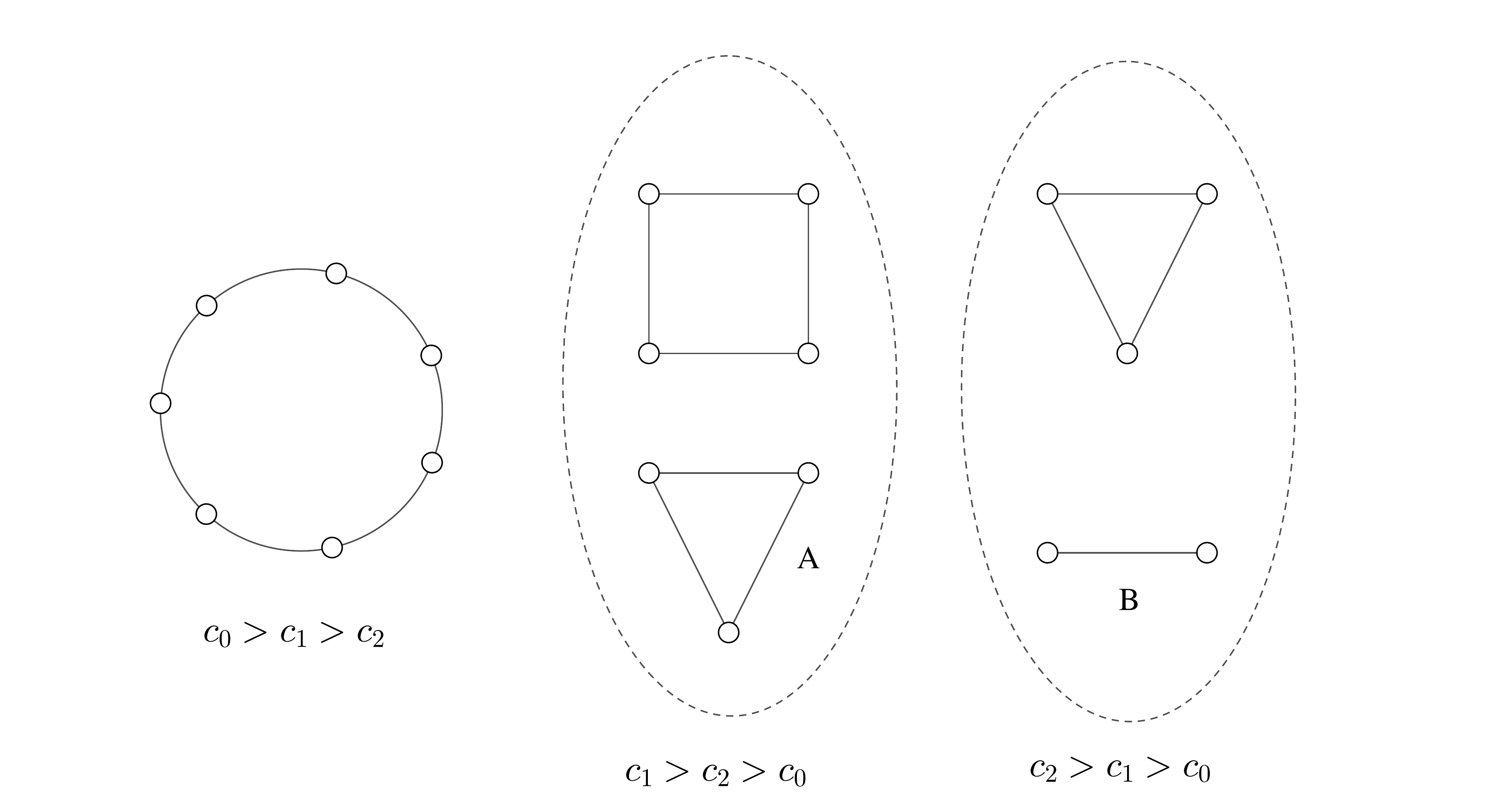}
	\caption{example of a small undirected network in which the greedy algorithm performs badly ($|A|=3$, $|B|=2$).}
	\label{fig:example1}
\end{figure}

Since, except for the nodes that already vote for her, $c_0$ is always ranked as third, it is clear that she cannot get any further vote. Then, in order to increase the margin of victory of $c_0$, we have that $c_2$ must obtain some of the $c_1$'s votes. The optimal solution $(S^*,I^*)$ is that, while $c_0$ keeps 7 votes, $c_1$ and $c_2$ collect 6 votes each, providing $\Expec_H \left[\MoV(S^*,I^*, H)\right] = 1$. This can be obtained by forcing $I(v)=(\cdot,\cdot,+)$ for a single $v \in A$ and $I(w)=(\cdot,+,\cdot)$ for a single $w \in B$.

However, this solution cannot be found by any algorithm adopting the greedy approach described above. Indeed, we next show that $\mathcal{F}(\emptyset,())$ is empty, and thus the algorithm never adds any seed in $S$: clearly, $\mathcal{F}(\emptyset,())$ cannot contain any pair $(s,I(s))$ that increases the number of votes of $c_0$; moreover,
by seeding a node in the seven-node ring the margin of victory clearly cannot increase (it either remains unchanged, or it decreases if $c_0$ ceases to be the best ranked candidate); similarly, by seeding a node in the four-node ring or in $A$, either the margin of victory goes down (if $c_2$ passes $c_1$) or remains unchanged; finally, by seeding one of the remaining nodes either the margin of victory goes down (if $c_1$ passes $c_2$) or remains unchanged.

Hence, the greedy solution results in a zero margin of victory, and thus it cannot be a $\rho$-approximation.
\end{proof}

\begin{proposition}
\label{prop:example2}
Be given the set of instances in which $(\phi, |C|)$ is not least-candidate manipulable.  For any $\rho > 0$ even depending logarithmically in the size of the problem, no algorithm following the greedy approach returns a $\rho$-approximation to the election problem,  even in directed trees.
\end{proposition}
\begin{proof}
Let $r > \frac{2}{\rho}$, and consider the graph given in Figure~\ref{fig:example2}.
 According to the ranking preferences of the nodes, the candidate $c_2$ collects $5r$ votes, while candidates $c_1$ and $c_0$ gather $7r$ votes each. Then, the actual margin of victory is equal to zero. As in the previous case, suppose  the budget $B$ be 2.

\begin{figure}[!htb]
	\centering
	\includegraphics[width=1.1\columnwidth]{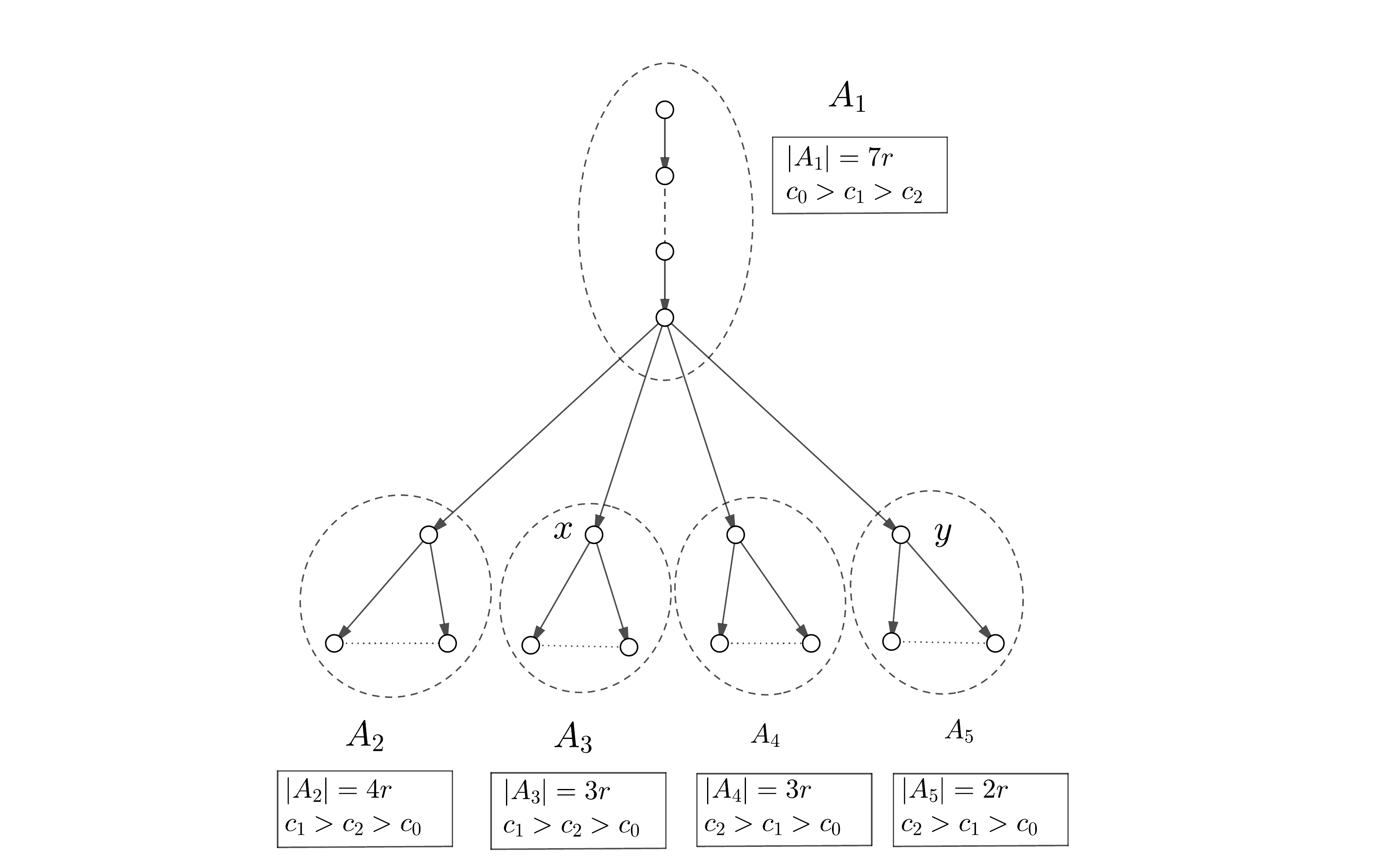}
	\caption{example of a tree in which the greedy algorithm performs badly.}
	\label{fig:example2}
\end{figure}

As above, we have that $c_0$ cannot raise any more vote, and we need that $c_2$ takes some of the $c_1$ votes. The optimal solution $(S^*,I^*)$ is then obtained forcing that $S^*=\{x, y\}$ with $I(x)=(\cdot,\cdot,+)$ and $I(y)=(\cdot,+,\cdot)$, and the expected margin of victory is $r$.

However, this solution cannot be found by any algorithm adopting the greedy approach. Indeed, as above, we have that no pair $(s,I(s))$ can increase the number of votes of $c_0$; moreover, by seeding a node in $A_1$ the margin of victory clearly cannot increase (it either remains unchanged, or it decreases if $c_0$ ceases to be the best ranked candidate); by seeding a node in $A_4$ or in $A_5$, either the margin of victory goes down (if $c_1$ passes $c_1$) or remains unchanged; finally, by seeding the root of $A_1$ or the root of $A_2$, either the margin of victory goes down (if $c_2$ passes $c_1$) or remains unchanged. Hence, the only action that the greedy algorithm can take would be to select as seed either a leaf of $A_1$, or a leaf of $A_2$ and letting them to change its vote from $c_1$ to $c_2$. By repeating the argument, we have that the two seeds selected by a greedy algorithm, must be two leafs from $A_1 \cup A_2$. So, the expected margin of victory is $2$, resulting in an approximation factor of $\frac{2}{r} < \rho$.
\end{proof}

\begin{remark}
 Another less natural class of greedy algorithms would dictate to choose the next seed so that the second most voted candidate loses some votes (even if this does not increase the margin of victory). It is not hard to see that the example deployed in Proposition~\ref{prop:example2} still proves that even this class of algorithms fails even for very simple graphs.
\end{remark}

We recall that greedy algorithms are essentially the only known algorithms guaranteeing bounded approximations for many problems related to the election control problem, such as the well-known influence maximization problem. Hence, even if an algorithm exists enabling the manipulator to control the election on many instances, Proposition~\ref{prop:example1} and Proposition~\ref{prop:example2} show that new approaches are necessary to design it.

\subsection{Approximation Results}
We next show that the condition behind the inapproximability results is tight. Indeed, by dropping that condition, we can design poly-time constant-approximation algorithms. Moreover, these algorithms turn out to follow the greedy approach that we proved to fail even for simple structure whenever the least-preferred candidate cannot be made to win.

Before presenting our result, let us introduce the following definition:
given a live graph $H$ and a choice of seeds $S$, the \emph{influence} $\chi(S, H)$ is defined to be the number of voters that are activated by these seeds: note that these are exactly the nodes that are reachable in the live-graph $H$ by at least one seed in $S$. In what follows, we  assume, for sake of presentation, that the expected influence $\Expec_H\left[\chi(S, H)\right]$ can be computed in poly-time. However, if this is not the case, we can still use a Monte Carlo simulation to approximate the expected influence within a factor $\gamma$, for every $\gamma > 0$. It turns out that, using such an approximation in place of the correct value for $\Expec_H\left[\chi(S, H)\right]$ will alter the approximation ratio of the proposed algorithms only for an additive factor $\varepsilon = \varepsilon(\gamma)$ as discussed by~\citeauthor{kempe2003maximizing}~[\citeyear{kempe2003maximizing}].
\begin{theorem}
\label{thm:approx}
Be given the set of instances in which $(\phi, |C|)$ is least-candidate manipulable.  Call $\tau \leq \ell +1$ the cardinality of the smallest set of messages making $c_0$ be the most preferred for every initial preference ranking. There is a greedy poly-time  algorithm returning a $\rho$-approximation to the election control problem, with
 $$
  \rho = \frac{B - \tau + 1}{2\tau B}\left(1-\frac{1}{e}\right).
 $$
\end{theorem}
\begin{proof}
Let $I^*$ be the set of $\tau$ messages that cause a voter to vote for $c_0$, whatever was the ranking before the reception of these messages. 
Notice that since the instance is least-candidate manipulable at least the set of messages $\{(c_0,+),(c_i,-) \textnormal{ for every $i >0$ } \}$ satisfies this property.
Our algorithm selects $\left\lfloor\frac{B}{\tau}\right\rfloor$ seeds through the greedy algorithm to maximize $\Expec_H[\chi(S,H)]$ (i.e., take at each time the seed that most increases this quantity), and let each of them to send all messages in $I^*$.
It directly follows that this algorithm runs in poly-time in greedy fashion.

In order to formally prove the approximation factor of this algorithm for the election control problem, let us denote
with $\hat{S}$ the set of seeds returned by our algorithm,
with $S^*$ the set of seeds maximizing $\Expec_H[\Delta_\MoV(S,I^*,H)]$,
with $S'$ the set of seeds of size $B$ that maximizes $\Expec_H[\chi(S,H)]$ and
with $S''$ the set of seeds of size $k = \left\lfloor\frac{B}{\tau}\right\rfloor$ that maximizes $\Expec_H[\chi(S,H)]$.

It is known that the function $\Expec_H[\chi(S, H)]$ is monotonic and submodular on $S$ \cite{kempe2003maximizing}, i.e., $\Expec_H[\chi(S,H)] \leq \Expec_H[\chi(T,H)]$ and $\Expec_H[\chi(S \cup \{x\}, H)] - \Expec_H[\chi(S, H)] \geq \Expec_H[\chi(T \cup \{x\}, H)] - \Expec_H[\chi(T, H)]$ for every $S \subseteq T$ and every $x \notin T$. Consequently, the greedy algorithm is known to return, for every $k$,
a set of $k$ seeds whose influence is an $\left(1 - \frac{1}{e}\right)$-approximation of the maximum expected influence achievable with $k$ seeds \cite{kempe2003maximizing}. Hence, we have that:
\begin{equation}
 \label{eq:proof1}
 \Expec_H[\chi(\hat{S},H)] \geq \left(1 - \frac{1}{e}\right) \Expec_H[\chi(S'',H)].
\end{equation}

Note that $|V_{c}|- \Expec_H\left[\left|V^*_c(S^*,I^*,H)\right|\right] \le  \Expec_H\left[\chi(S^*,H)\right]$ for every $c \neq c_0$, since at most one vote can be lost by $c$ for every influenced node in graph $H$.
Then we have
\begin{equation}
\label{eq:proof2}
\begin{aligned}
 & \max_{c \neq c_0} |V_{c}|- \Expec_H\left[\max_{c \neq c_0} \left|V^*_c(S^*,I^*,H)\right|\right]\\
 & \qquad \le \max_{c \neq c_0} \Big\{|V_{c}|- \Expec_H\left[\left|V^*_c(S^*,I^*,H)\right|\right] \Big\}\\
 & \qquad \le \Expec_H\left[\chi(S^*,H)\right].
\end{aligned}
\end{equation}

A similar argument proves that
\begin{equation}
 \label{eq:proof3}
 \Expec_H\left[\left|V^*_{c_0}(S^*,I^*,H)\right|\right] - |V_{c_0}| \le \Expec_H\left[\chi(S^*,H)\right].
\end{equation}

Moreover, $\frac{\Expec_H\left[\chi(S',H)\right]}{|S'|} \leq \frac{\Expec_H\left[\chi(S',H)\right]}{|S''|}$, by submodularity of $\chi$.
Since $|S'| = B$ and $|S''| = \left\lfloor\frac{B}{\tau}\right\rfloor \geq \frac{B-\tau+1}{\tau}$, we then achieve that
\begin{equation}
\label{eq:proof4}
 \Expec_H\left[\chi(S',H)\right] \leq \frac{\tau B}{B-\tau+1}\Expec_H\left[\chi(S'',H)\right].
\end{equation}

Moreover, by definition of $\Delta_\MoV$,
$\Expec_H[\Delta_\MoV(S^*,I^*,H)] = (\Expec_H\left[\left|V^*_{c_0}(S^*,I^*,H)\right|\right] - \Expec_H\left[\max_{c \neq c_0} \left|V^*_c(S^*,I^*,H)\right|\right]) - \left(|V_{c_0}| - \max_{c \neq c_0} |V_{c}| \right)$. Hence, we directly achieve that
$\Expec_H[\Delta_\MoV(S^*,I^*,H)] = (\Expec_H\left[\left|V^*_{c_0}(S^*,I^*,H)\right|\right] - |V_{c_0}|) + (\Expec_H[\max_{c \neq c_0} |V^*_c(S^*,I^*,H)|] - \max_{c \neq c_0} |V_{c}|)$.
Putting all together, we then have that
\begin{align*}
 & \Expec_H[\Delta_\MoV(S^*,I^*,H)]\\
 & \qquad \leq 2\Expec_H\left[\chi(S^*,H)\right]  \hspace{3cm} \text{(by \eqref{eq:proof2} and \eqref{eq:proof3})}\\
 & \qquad \leq 2\Expec_H\left[\chi(S',H)\right] \hspace{2.3cm} \text{(by definition of $S'$)}\\
 & \qquad \leq \frac{2\tau B}{B-\tau+1}\Expec_H\left[\chi(S'',H)\right] \hspace{2.54cm} \text{(by \eqref{eq:proof4})}\\
 & \qquad \leq \frac{2\tau B}{B-\tau+1} \left(1-\frac{1}{e}\right)^{-1}\Expec_H\left[\chi(\hat{S},H)\right] \hspace{0.7cm} \text{(by \eqref{eq:proof1})}\\
 & \qquad \leq \frac{2\tau B}{B-\tau+1} \left(1-\frac{1}{e}\right)^{-1}\Expec_H\left[\Delta_\MoV(\hat{S},I^*,H)\right],
\end{align*}
where the last inequality follows from the fact that, by definition of $I^*$, all influenced nodes will vote for $c_0$.
\end{proof}

\subsection{Special Cases}
We next show some specific results for the special ranking revision functions discussed in Section~\ref{sec:model}. 
Since pessimistic ranking revision function is not least-candidate manipulable for three candidates, then Theorem~\ref{thm:inapprox} applies, and thus we have the following.
\begin{proposition}
\label{prop:pessimistic3}
 For every $\rho >0$ even depending on the size of the problem, unless $\mathsf{P} = \mathsf{NP}$ there is no poly-time algorithm returning a $\rho$-approximation to the election problem with pessimistic ranking revision function, for only three candidates.
\end{proposition}

By focusing on optimistic ranking revision function, we have that two negative messages about $c_1$ and $c_2$, i.e. $(\cdot,-,-)$, are sufficient to make $c_0$ first for every initial ranking.
Hence, we can apply Theorem~\ref{thm:approx} with $\tau = 2$, and we have the following.
\begin{proposition}
 \label{prop:optimistic3}
 There is a greedy poly-time algorithm that returns a $\rho$-approximation to the election control problem with optimistic ranking revision function for three candidates, with
 $$
  \rho = \frac{B - 1}{4B}\left(1-\frac{1}{e}\right) \approx \frac{1}{4} \left(1-\frac{1}{e}\right).
 $$
 The result holds even if only negative influences are allowed.
\end{proposition}

Differently from the optimistic ranking revision function, for the score-based one, $(\cdot,-,-)$ is not sufficient for make $c_0$ first when she is last in the initial ranking. For this reason, to ensure that $c_0$ reaches the first position, three messages are needed, namely $(+,-,-)$. Hence, from Theorem~\ref{thm:approx} we achieve the following.
\begin{proposition}
	\label{prop:score3}
	There is a greedy poly-time algorithm that returns a $\rho$-approximation to the election control problem with score-based ranking revision function for three candidates, with
	$$
	\rho = \frac{B - 2}{6B}\left(1-\frac{1}{e}\right) \approx \frac{1}{6} \left(1-\frac{1}{e}\right).
	$$
\end{proposition}

 Moreover, differently from the optimistic ranking revision function, for the score-based one the approximation guarantee ceases to hold if only negative influences are allowed. Indeed, in this case there is no way for the last ranked candidate to become the first one, and thus Theorem~\ref{thm:inapprox} holds.
\begin{proposition}
	\label{prop:score3neg}
	For every $\rho >0$ even depending on the size of the problem, unless $\mathsf{P} = \mathsf{NP}$, then, even for only three candidates, there is no poly-time algorithm returning a $\rho$-approximation to the election problem with score-based ranking revision function if only negative influences are allowed.
\end{proposition}

We show that, whereas an approximation algorithm exists for the optimistic ranking revision function with only negative influence, this result does not carry on to the case that only positive influence is allowed. Indeed, in this case, there is no way to raise the rank of the last candidate up to the first position, and thus Theorem~\ref{thm:inapprox} applies.

\begin{proposition}
\label{prop:opt3pos}
 For every $\rho >0$ even depending on the size of the problem, unless $\mathsf{P} = \mathsf{NP}$, then, even for only three candidates and optimistic ranking revision function, there is no polynomial time algorithm returning a $\rho$-approximation to the election problem if only positive influences are allowed.
\end{proposition}

To conclude, we extend the above results to the case of score-based tie-breaking rule with four or more candidates. We recall that this tie-breaking rule is well-defined for every number of candidates. 
More precisely, it turns out that having more candidates makes the problem even harder. Indeed, with more than three candidates, there is no way of pushing the last ranked candidate up to the first rank, even with score-based tie-breaking rules, and hence Theorem~\ref{thm:inapprox} holds.
\begin{proposition}
\label{prop:score4neg}
 For every $\rho >0$ even depending on the size of the problem, if there are at least four candidates, then, unless $\mathsf{P} = \mathsf{NP}$, there is no poly-time algorithm returning a $\rho$-approximation to the election problem with score-based tie-breaking rule.
\end{proposition}

\section{Extensions}
We describe some extensions and variants of our model and  show how most of our results extend to these settings.
\subsection{Bribed Voters}
In the model described in Section~\ref{sec:model}, seeds act as initiators of positive and negative messages about candidates. However, apart from that, their behavior is exactly the same as any other node in the network. In particular, the messages that they receive will affect their ranking and, consequently, their vote.
We now consider also a variant of this scenario, in which seeds are \emph{bribed}, i.e., for each seed her preference ranking (and thus, her vote) is not affected by messages different from the one that she sends.

It directly follows that the reduction described in the proof of Theorem~\ref{thm:inapprox} does not work in this setting.
However, we next show that, even in this setting, the election control problem is essentially inapproximable.
\begin{theorem}
\label{thm:bribed_inapprox}
Be given the set of instances in which the pairs composed of tie-breaking rule and number of candidates are not least-candidate manipulable. For any $\rho > 0$  there is no poly-time algorithm returning a $\rho$-approximation to the election problem with bribed seeds, unless $\mathsf{P} = \mathsf{NP}$.
\end{theorem}

\begin{proof}[Proof Sketch]
 Consider the reduction described in the proof of Theorem~\ref{thm:inapprox}, except that now each node is enlarged into a clique of size $(h+1) \rho'$, where $\rho' > \rho$.
 Hence, if a set cover of size at most $h$ exists, then, $\Delta_{\MoV}(S^*,I^*,H) \geq (h+1) \rho'$, otherwise the only nodes that eventually change opinion are the seeds, that are at most $h+1$. Thus any $\rho$-approximate algorithm must be able to distinguish these two cases and thus solves the set cover problem in poly-time.
\end{proof}
Instead, it is easy to check that Theorem~\ref{thm:approx} is unaffected by bribed voters, and so a constant approximation is still possible when there is a set of messages able to lead the last ranked candidate to the first place.

\subsection{Other Objective Functions}
In addition to the maximization of the increase in the margin of victory, also studied by~\citeauthor{wilder2018controlling}~[\citeyear{wilder2018controlling}], alternative objective functions may be of interest.

For example, one may want to maximize the probability of victory. For this objective function, already discussed by~\citeauthor{wilder2018controlling}~[\citeyear{wilder2018controlling}], it is not trivial to see that Theorem~\ref{thm:inapprox} keeps holding.
However, notice that this objective function makes the problem even harder than maximizing the margin of victory. Indeed, for the latter objective, Theorem~\ref{thm:approx} implies that a $\frac{1}{2}\left(1 - \frac{1}{e}\right)$-approximation can be computed in poly-time when only two candidates are involved. It is instead not hard to see that, in order to maximize the probability of victory when only two candidates are equivalent, it is sufficient that all selected seeds send the same message. Hence, for two candidates, maximizing the probability of victory in our setting is the same as doing it in the setting of~\citeauthor{wilder2018controlling}~[\citeyear{wilder2018controlling}]. Hence, the problem cannot be approximate, unless $\mathsf{P}=\mathsf{NP}$, within a factor $\rho > 0$, even for two only candidates.

An apparently weaker goal would be that one of computing the set of seeds and the corresponding messages so that the probability of victory merely is above a given threshold (so the set of feasible solutions would be larger than in the setting described above). Unfortunately, this objective function does not make the problem easier to be solved. Indeed, not only Theorem~\ref{thm:inapprox} holds in this setting regardless than the threshold, but one may show that, as for the goal of maximizing the probability of victory, the inapproximability still holds when only two candidates are available.
 
\subsection{Threshold Dynamics}
The results we derived in the previous sections and based to a multi-issue independent cascade model can be extended to settings in which the diffusion model is linear threshold. This model  represents the most prominent among the diffusion models alternative to the independent cascade. In the linear threshold model, 
for each node $v$ of the network, there is a threshold $\theta_v$ drawn randomly in $[0,1]$, and incoming edges $(u,v)$ have a weight $w_{u,v}$ such that $\sum_{(u,v)} w_{u,v} = 1$.
Then, a node $v$ becomes active at time $t$ only if the sum of weights of edges coming from active nodes passes the threshold.

It is easy to check that this diffusion model leads to different dynamics with respect to the independent cascade model. Still, we show that our proofs can be adapted. In particular,  Theorem~\ref{thm:inapprox} and Theorem~\ref{thm:approx} still hold.

Specifically, for the inapproximability result, we use, in place of set cover, a reduction from vertex cover.
This is the problem of deciding whether, given a graph $Z$ of $m$ nodes and an integer $k$, there is a subset $S$ of at most $k$ nodes of $Z$ such that every edge of $Z$ has at least one endpoint in $S$.
The reduction  is similar to the one described in Theorem~\ref{thm:inapprox}. Namely, the component $G_1$ consists of the graph $G$. Now by setting $n =m-k$, we let components $G_2-G_4$ to have the same number of nodes as in the proof of Theorem~\ref{thm:inapprox}, except that now the nodes in each components are not arranged as a clique, but as a directed ring (so that a message sent by a node in one component will activate all nodes in that component regardless of their threshold). Notice that, by considering the same initial rankings as in the proof of Theorem~\ref{thm:inapprox},  the margin of victory of $c_0$ in this instance can become greater than 0 if and only if there is in $G$ a vertex cover of size at most $k$.

On the other side, it directly follows that the greedy algorithm proposed in Theorem~\ref{thm:approx} works, with the same approximation factor, even with the linear threshold diffusion model. Indeed, it is known that the influence maximization is a monotonic and submodular function even with this dynamics \cite{kempe2003maximizing}. However, it can be observed that this is sufficient to make the proof of Theorem~\ref{thm:approx} hold.

\subsection{Seeds with Different Costs}
In our model we assume that each node can be selected as a seed at same cost. This can be highly unrealistic. Hence, an extension to our model would be to assume that each node $u$ has a different cost $w(u)$ that should be payed for each message initiated by that node.

Intuitively, this extension makes the election control problem harder.
Hence, inapproximability results clearly extend to this setting too.
Interestingly, however, we have that, whenever there are messages such that the last ranked candidate can be driven to the first position, a poly-time algorithm returning a constant approximation to the electoral control problem exists even if nodes have heterogeneous seeding costs.
Indeed, it is known that in this setting there is a poly-time algorithm for influence maximization returning a $\left(1 - \frac{1}{\sqrt{e}}\right)$-approximation of the optimal seed set \cite{nguyen2013budgeted}.
Then, the arguments of the proof of Theorem~\ref{thm:approx} immediately prove that this algorithm
provides a $\rho$-approximation for the extension of the election control problem to voters with different costs, where
 $\rho = \frac{B - \tau + 1}{2\ell B}\left(1-\frac{1}{\sqrt{e}}\right)$.

\section{Conclusions and Future Work}
In this work, we analyzed the problem of manipulating the result of an election by seeding the network with messages both positive and negative towards the candidates. We prove a tight characterization of the settings in which computing an approximation to the best manipulation can be infeasible or feasible. Specifically, we show that a transition phase exists when we pass from a setting where the last ranked candidate for a voter can never be promoted to the most-preferred one to the settings in which this promotion is allowed, and this holds regardless the ranking revision function one can use.
We both applied these results to special cases and showed how to generalize them to hold even for variants and/or generalizations of the original model. We also show that, in simple networks, a large class of algorithms, that mainly include all approaches recently adopted for social-influence problems, fail to compute an empirically bounded approximation even on very simple networks, as undirected graphs with every node having a degree at most two or directed trees.

Nevertheless, we believe that our results can be  refined to have a more detailed picture of the problem.
For example, our inapproximability results are achieved by assuming that the underlying network is directed.
It would then be interesting to understand whether these results extend to undirected graphs or the latter embeds features that can be exploited by a manipulator even if the last ranked candidate cannot be promoted to the first position. We note that an adaption the proof of Theorem~\ref{thm:inapprox} along the line of the reduction for influence maximization described by \citeauthor{khanna2014influence}~[\citeyear{khanna2014influence}] can be used to prove partial results in this direction, namely that inapproximability holds even for undirected graphs whenever we have both heterogeneous seeding costs and the weight of each vote changing from voter to voter.
More in general, a finer characterization of the network structure for which manipulation is hard/easy would be of extreme interest.

While we provided a poly-time constant-approximation algorithm in many settings, we did not try to optimize the approximation ratio. Hence, it would be interesting to design algorithms that can  improve on ours. Finally, it would be interesting to analyze other generalizations of our model, e.g., different models for information diffusion and time-evolving networks.

\bibliographystyle{named}
\bibliography{elections}


\end{document}